\documentclass[11pt]{article}
\usepackage{amsmath,amssymb,amsthm}
\usepackage{amsfonts,amssymb}
\usepackage{dsfont}
\usepackage[colorlinks=true, linkcolor=red, urlcolor=blue, citecolor=gray]{hyperref}
\usepackage{color,soul}
\usepackage{xcolor}
\usepackage{mdframed}
\usepackage{comment}
\usepackage{graphicx}
\usepackage{mathtools}
\usepackage[export]{adjustbox}
\usepackage{algorithm}
\usepackage[noend]{algpseudocode}
\usepackage[normalem]{ulem}
\usepackage{fancybox}
\usepackage{tkz-graph}
\usetikzlibrary{decorations.pathmorphing}
\usepackage{subcaption}

\usepackage{tikz}

\definecolor{green2}{RGB}{34,139,34}
\definecolor{green}{RGB}{34,139,34}

	\newcommand{\ignore}[1]{}

\newcommand{\eps}{\varepsilon}

\newcommand{\R}{\mathbb{R}}

\DeclareMathOperator{\polylog}{polylog}

\DeclarePairedDelimiter{\ceil}{\lceil}{\rceil}

\newcommand{\e}{\epsilon}

\newtheorem{theorem}{Theorem}

\newtheorem{lemma}{Lemma}
\newtheorem{observation}{Observation}

\newtheorem{definition}{Definition}

\newtheorem{claim}{Claim}

\newtheorem{fact}{Fact}

\newtheorem{example}{Example}
\newcommand{\poly}{\text{poly}}
\newcommand{\wt}{\widetilde}
\newcommand{\T}{\top}
\newcommand{\ape}{\approx_\e}

\newenvironment{proofof}[1]{\noindent{\bf Proof of #1:}}{$\qed$\par}

\advance \topmargin by -\headheight
\advance \topmargin by -\headsep
\evensidemargin \oddsidemargin
{\makeatletter
 \gdef\xxxmark{%
   \expandafter\ifx\csname @mpargs\endcsname\relax 
     \expandafter\ifx\csname @captype\endcsname\relax 
       \marginpar{xxx}
     \else
       xxx 
     \fi
   \else
     xxx 
   \fi}
 \gdef\xxx{\@ifnextchar[\xxx@lab\xxx@nolab}
 \long\gdef\xxx@lab[#1]#2{{\bf [\xxxmark #2 ---{\sc #1}]}}
 \long\gdef\xxx@nolab#1{{\bf [\xxxmark #1]}}
}

\newcommand{\plog}{\mathop\mathrm{polylog}}

\newcommand{\norm}[1]{\|#1\|}

\title{Dynamic Streaming Spectral Sparsification\\ 
in Nearly Linear Time and Space}
\author{Michael Kapralov\\EPFL \and  Navid Nouri \\EPFL \and Aaron Sidford  \\Stanford University  \and Jakab Tardos \\EPFL}

  \usepackage{nth}
  \usepackage{intcalc}

  \newcommand{\cSTOC}[1]{\nth{\intcalcSub{#1}{1968}}\ Annual\ ACM\ Symposium\ on\ Theory\ of\ Computing\ (STOC)}

  \newcommand{\cICALP}[1]{\nth{\intcalcSub{#1}{1973}}\ International\ Colloquium\ on\ Automata,\ Languages and\ Programming\ (ICALP)}

  \newcommand{\cPODS}[1]{\nth{\intcalcSub{#1}{1981}}\ Symposium\ on\ Principles\ of\ Database\ Systems\ (PODS)}

  \newcommand{\pICALP}[1]{Preliminary\ version\ in\ the\ \cICALP{#1}, #1}

  \newcommand{\STOC}[1]{Proceedings\ of\ the\ \cSTOC{#1}}

  \newcommand{\PODS}[1]{Proceedings\ of\ the\ \cPODS{#1}}

  \newcommand{\arXiv}[1]{\href{http://arxiv.org/abs/#1}{arXiv:#1}}

\begin{document}

\maketitle

\begin{abstract}
In this paper we consider the problem of computing {\em spectral approximations} to graphs in the single pass dynamic streaming model. We provide a linear sketching based solution that given a stream of edge insertions and deletions to a $n$-node undirected graph, uses $\tilde O(n)$ space, processes each update in $\tilde O(1)$ time, and with high probability recovers a spectral sparsifier in $\tilde O(n)$ time. Prior to our work, state of the art results either used near optimal $\tilde  O(n)$ space complexity, but brute-force $\Omega(n^2)$ recovery time [Kapralov et al.'14], or with subquadratic runtime, but polynomially suboptimal space complexity [Ahn et al.'14, Kapralov et al.'19].

Our main technical contribution is a novel method for `bucketing' vertices of the input graph into clusters that allows fast recovery of edges of sufficiently large effective resistance. Our algorithm first buckets vertices of the graph by performing ball-carving using (an approximation to) its effective resistance metric, and then recovers the high effective resistance edges from a sketched version of an electrical flow between vertices in a bucket, taking nearly linear time in the number of vertices overall. This process is performed at different geometric scales to recover a sample of edges with probabilities proportional to effective resistances and obtain an actual sparsifier of the input graph. 

This work provides both the first efficient $\ell_2$-sparse recovery algorithm for graphs and new primitives for manipulating the effective resistance embedding of a graph, both of which we hope have further applications.
\end{abstract}


\newpage

\section{Introduction}
Graph sketching, i.e. constructing small space summaries for graphs using linear measurements, has received much attention since the work of Ahn, Guha and McGregor~\cite{ahn2012analyzing} gave a linear sketching primitive for graph connectivity with optimal $O(n\log^3 n)$ space complexity~\cite{NelsonY19}.  
A key application of linear sketching has been to design small space algorithms for processing {\em dynamic graph streams}, where edges can be both inserted and deleted, although the graph sketching paradigm has been shown very powerful in many other areas such as distributed algorithms and dynamic algorithms (we refer the reader to the survey~\cite{McGregor17} for more on applications of graph sketching).   Furthermore, it is known that linear sketching  is essentially a universal approach to designing dynamic streaming algorithms~\cite{LiYiWoodruff14}, and yields distributed protocols for graph processing with low communication. Sketching solutions have been recently constructed for many graph problems, including spanning forest computation~\cite{ahn2012analyzing}, cut and spectral sparsifiers~\cite{AhnGM12,streamingSpectral}, spanner construction~\cite{AhnGM12,KapralovW14}, matching and matching size approximation~\cite{AssadiKLY16,AssadiKL17}, sketching the Laplacian~\cite{AndoniCKQWZ16,JambulapatiS18} and many other problems.  The focus of our work is on {\em oblivious} sketches for approximating spectral structure of graphs with optimally fast recovery. A sketch is called {\em oblivious} if its distribution is independent of the input -- such sketches yield efficient {\em single pass} dynamic streaming algorithms for sparsification. We now outline the main ideas involved in previous works on this and related problems, and highlight the main challenges in designing a solution that achieves both linear space and time. 
 
 Oblivious linear sketches with nearly optimal $n\log^{O(1)} n$ have been obtained for the related problems of constructing a spanning forest of the input graph~\cite{ahn2012analyzing}, the problem of constructing {\em cut sparsifiers} of graphs~\cite{ahn2012graph} and for the spectral sparsification problem itself~\cite{streamingSpectral}. In the former two cases the core of the problem is to design a sketch that allows recovery of edges that cross {\em small cuts} in the input graph, and the problem is resolved by applying $\ell_0$-sampling(see, e.g.,~\cite{JowhariST11,CormodeF14,KapralovNPWWY17}), and more generally exact (i.e., $\ell_0$) sparse recovery techniques on the edge incidence matrix $B\in \R^{{n \choose 2}\times n}$ of the input graph: one designs a sketching matrix $S\in \R^{\log^{O(1)} n\times {n \choose 2}}$ and maintains $S\cdot B\in \R^{\log^{O(1)} n \times n}$ throughout the stream. A natural recovery primitive that follows Boruvka's algorithm for the MST problem then yields a nearly linear time recovery scheme.  Specifically, to recover a spanning tree one repeatedly samples outgoing edges out of every vertex of the graph and contracts resulting connected components into supernodes, halving the number of connected components in every round. Surprisingly, a sketch of the original graph suffices for sampling edges that go across connected components in graphs that arise through the contraction process, yielding a spanning forest in $O(\log n)$ rounds and using $n\log^{O(1)} n$ bits of space.
 
 The situation with spectral sparsifiers is very different: edges critical to obtaining a spectral approximation do not necessarily cross small cuts in the graph. Instead, `important edges' are those that have large effective resistance, i.e can be made `heavy' in the $\ell_2$ sense in an appropriate linear combination of the columns of the edge incidence matrix $B$. This observation was used in~\cite{streamingSpectral} to design a sketch with nearly optimal $n\log^{O(1)} n$ space complexity, but the recovery of the sparsifier was brute-force and ran in $\Omega(n^2)$ time: one had to iterate over all potential edges and test whether they are in the graph and have `high' effective resistance. Approaches based on relating effective resistances to inverse connectivity have been proposed~\cite{ahn2013spectral}, but these result in suboptimal $\Omega(n^{5/3})$ space complexity. In a very recent work~\cite{KMMMN19} a subset of the authors proposed an algorithm with $n^{1.4+o(1)}$ space and runtime complexity, but no approach that yields optimal space and runtime was known previously. 
 
A key reason why previously known sketching techniques for reconstructing spectral approximations to graphs failed to achieve nearly linear runtime is exactly the lack of simple `local' (akin to Boruvka's algorithm) technique for recovering heavy edges. The main contribution of this paper is such a technique: we propose a bucketing technique based on ball carving in (an approximation to) the effective resistance metric that  recovers appropriately heavy effective resistance edges by routing flows between source-sink pairs that belong to the same bucket. This ensures that the recovery process is more `localized', and results in a nearly linear time algorithm.

\paragraph{Our result.} Formally, we consider the problem of constructing {\em spectral sparsifiers}~\cite{SpielmanTengSparsification,spielman2011graph} of graphs presented as a dynamic stream of edges:  given a graph $G=(V, E)$ presented as a dynamic stream of edge insertions and deletions and a precision parameter $\e \in (0, 1)$, our algorithm outputs a graph $G'$ such that
$$
(1-\e) L\preceq L'\preceq (1+\e)L,
$$
where $L$ is the Laplacian of $G$, $L'$ is the Laplacian of $G'$ and $\prec$ stands for the positive semidefinite ordering of matrices.

Our main result is a linear sketching algorithm that compresses a graph with $n$ vertices to a $n \log^{O(1)} n$-bit representation that allows $\log^{O(1)} n$-time updates, and from which a spectral approximation can be recovered in $n\log^{O(1)} n$ time. Thus, our result achieve both optimal space and time complexity simultaneously. 

\begin{theorem}[Near Optimal Streaming Spectral Sparsification] \label{thm:main}
There exists an algorithm such that for any $\epsilon\in (0, 1)$, processes a list of edge insertions and deletions for an unweighted graph $G$ in a single pass and maintains a set of linear sketches of this input in $O(\e^{-2}n\log^{O(1)} n)$ space. From these sketches, it recovers in $O(\e^{-2}n\log^{O(1)} n)$ time, with high probability, a weighted subgraph $H$ with $O(\epsilon^{-2}n\log n)$ edges, such that $H$ is a $(1 \pm \epsilon)$-spectral sparsifier of $G$.
\end{theorem}

Our result in Theorem~\ref{thm:main} can be thought of as the first efficient `$\ell_2$-graph sketching' result, using an analogy to compressed sensing recovery guarantees.  It is interesting to note that compressed sensing primitives that allow recovery in time nearly linear in sketch size (which is exactly what our algorithm achieves for the sparsification problem) usually operate by hashing the input vector into buckets so as to isolate dominant entries, which can then be recovered efficiently. The main contribution of our work is giving a `bucketing scheme' for graphs that allows for nearly linear time recovery. As we show, the right `bucketing scheme' for the spectral sparsification problem is a space partitioning scheme in the effective resistance metric.

\paragraph{Effective resistance, spectral sparsification, and random spanning trees.} 
The \emph{effective resistance metric} or \emph{effective resistance distances} induced by an undirected graph plays a central role in spectral graph theory and has been at the heart of numerous algorithmic breakthroughs over the past decade. They are central to the to obtaining fast algorithms for constructing spectral sparsifiers \cite{spielman2011graph, KoutisLP16}, spectral vertex sparsifiers \cite{KyngLPSS16}, sparsifiers of the random walk Laplacian \cite{ChengCLPT15, JindalKPS17}, and subspace sparsifiers \cite{LiS18}. They have played a key role in many advances in solving Laplacian systems \cite{SpielmanT04,KoutisMP10,KoutisMP11,PengS14,CohenKMPPRX14,KoutisLP16,KyngLPSS16,KyngS16} and are critical to the current fastest (weakly)-polynomial time algorithms for maximum flow and minimum cost flow in certain parameter regimes \cite{LeeS14}. Given their utility, the computation of effective resistances has itself become an area of active research \cite{JambulapatiS18,ChuGPSSW18}.

In a line of work particularly relevant to this paper, the effective resistance metric has played an important role in obtaining faster algorithms for generating random spanning trees~\cite{KelnerM09,MadryST15,Schild18}. The result of~\cite{MadryST15} partitions the graph into clusters with bounded diameter in the effective resistance metric in order to speed up simulation of a random walk, whereas~\cite{Schild18} proposed a more advanced version of this approach to achieve a nearly linear time simulation. While these results seem superficially related to ours, there does not seem to be any way of using spanning tree generation techniques for our purpose. The main reason is that the objective in spanning tree generation results is quite different from ours: there one would like to find a partition of the graph that in a sense minimizes the number times a random walk  crosses cluster boundaries, which does not correspond to a way of recovering `heavy' effective resistance edges in the graph. In particular, while in spanning tree generation algorithms the important parameter is the number of edges crossing the cuts generated by the partitioning, whereas it is easily seen that heavy effective resistance edges cannot be recovered from small cuts. Finally, the problem of partitioning graphs into low effective resistance diameter clusters has been studied recently in \cite{alev2017graph}. The focus of the latter work is on partitioning into {\em induced} expanders, and the results of~\cite{alev2017graph} were an important tool in the work of~\cite{KMMMN19} that achieved the previous best $n^{1.4+o(1)}$ space and runtime complexity for our problem.  Our techniques in this paper take a different route and achieve optimal results.

\paragraph{Prior work.} Streaming algorithms are well-studied with too many results to list and we refer the reader to \cite{McGregor14,McGregor17} for a survey of streaming algorithms. The idea of linear graph sketching was introduced in a seminal paper of Ahn, Guha, and McGregror \cite{ahn2012analyzing}, where a $O(\log n)$-pass sparsification algorithm for dynamic streams was presented (this result is for the weaker notion of cut sparsification due to~\cite{Karger94, BenczurKarger:1996}). A single-pass algorithm for cut sparsification with nearly optimal $\widetilde O(\e^{-2} n)$ space was given in~\cite{ahn2012graph}, and extensions of the sketching approach of~\cite{ahn2012analyzing} to hypergraphs were presented in~\cite{GuhaMT15}. The more challenging problem of computing a spectral sparsifier from a linear sketch was addressed in \cite{ahn2013spectral}, who gives an $\tilde{O}(\e^{-2} n^{5/3})$ space solution. An $\tilde{O}(\e^{-2} n)$ space solution was obtained in \cite{kapralov2017single} by more explicitly exploiting the connection between graph sketching and vector sparse recovery, at the expense of $\wt{O}(\e^{-2} n^2)$ runtime.  In a recent work~\cite{KMMMN19} a subset of the authors gave a single pass algorithm with $\e^{-2} n^{1.4+o(1)}$ space and runtime complexity. 

We also mention that spectral sparsifiers have been studied in the insertion-only streaming model, where edges can only be added to $G$ \cite{KelnerLevin:2013,cohen2016online,kyng2017framework}, and in a dynamic data structure model \cite{dynamicSparsifiers,AndoniCKQWZ16,jambulapati2018efficient}, where more space is allowed, but the algorithm must quickly output a sparsifier at every step of the stream. While these models are superficially similar to the dynamic streaming model, they seem to allow for different techniques, and in particular do not require linear sketching since they do not constrain the space used by the algorithm. The spectral sparsification problem on its own has received a lot of attention in the literature (e.g.,~\cite{spielman2011graph,SpielmanT11,BatsonSS09,ZhuLO15,LeeS15a,Lee017}. We refer the reader to the survey~\cite{BatsonSST13} for a more complete set of references.

\section{Preliminaries}
\label{sec:prelim}

\noindent{\textbf{General Notation.}} 
Let $G=(V,E)$ be an unweighted undirected graph with $n$ vertices and $m$ edges. For any vertex $v\in V$, let $\chi_v \in \R^n$ be the indicator vector of $v$, with a one at position $v$ and zeros elsewhere. Let $B_n\in \R^{\binom{n}{2}\times n}$ denote the vertex edge incidence matrix of an unweighted and undirected complete graph, where for any edge $e=(u,v)\in {V\choose 2}, u\ne v$, its $e$'th row is equal to $\mathbf{b}_e:= \mathbf{b}_{uv} := \chi_u-\chi_v$.  
Let $B\in \R^{\binom{n}{2}\times n}$ denote the vertex edge incidence matrix of $G=(V,E)$. $B$ is obtained by zeroing out any rows of $B_n$ corresponding to $(u,v)\notin E$.\footnote{Note this is different then the possibly more standard definition of $B$ as the $E \times V$ matrix with the rows not in the graph removed altogether.} 

For weighted graph $G=(V,E,w)$, where $w:E\rightarrow \R_+$ denotes the edge weights, let $W\in {\R_{+}}^{\binom{n}{2} \times \binom{n}{2} }$ be the diagonal matrix of weights where $W(e, e) = w(e)$ for $e \in E$ and $W(e,e) = 0$ otherwise. Note that $L=B^\T W B = B_n^T W B_n$, is the Laplacian matrix of $G$. Let $L^{+}$ denote the
Moore-Penrose pseudoinverse of $L$. Also, for a real valued variable $s$, we define $s^+:=\max\{0,s\}$. We also use the following folklore:
\begin{fact}\label{fact:lambda}
	For any Laplacian matrix $L$ of an unweighted and undirected graph, its minimum nonzero eigenvalue is bounded from below by $\lambda_\ell=\frac{1}{8n^2}$ and its maximum eigenvalue is bounded from above by $\lambda_u=2n$.
\end{fact}

\begin{definition}\label{def:Ggamma}
	For any unweighted graph $G=(V,E)$ and any $\gamma\ge0$, we define $L_{G^\gamma}$, as follows:
	$$L_{G^\gamma}=L_G+\gamma I.$$
	This can be seen in the following way. One can think of $G^{\gamma}$ as graph $G$ plus some regularization term. In order to distinguish between edges of $G$ and regularization term in $G^\gamma$, we let $B_{G^\gamma}=B \oplus \sqrt{\gamma}I$, where $B \oplus \sqrt{\gamma}I$ is the operation of appending rows of $\sqrt{\gamma}I $ to matrix $B$. One should note that $B_{G^\gamma}^\top B_{G^\gamma}=L_{G^\gamma}$. Also for simplicity we define $L_\ell$ for any integer $\ell \in [0,d+1]$ as follows:
	\begin{align*}
	L_{\ell}=\begin{cases}
	L_G+\frac{\lambda_u}{2^\ell}I & \text{if $0\le\ell\le d$}\\
	L_G & \text{if $\ell=d+1.$}
	\end{cases}
	\end{align*}
	where $d$ and $\lambda_u$ are defined as in Lemma~\ref{lem:chain_coarse}.
\end{definition}
We often denote the matrix $L_{G^\gamma}=L_G+\gamma I$ by $K$, and in particular use the notation $L$ and $K$ interchangeably.

\noindent{\textbf{Effective Resistance.}} 
Given a weighted graph $G =(V, E, w)$ we associate it with an electric circuit where the vertices are junctions and each edge $e$ is a resistor of resistance $1/w(e)$. Now suppose in this circuit we inject one unit current at vertex $u$, extract one from vertex $v$, and let $\mathbf{f}_{uv}\in \R^{m}$ denote the the currents induced on the edges. By Kirchhoff's current law, except for the source $u$ and the sink $v$, the sum of the currents entering and exiting any vertex is zero. Hence, we have $\mathbf{b}_{uv}=B^\T\mathbf{f}_{uv}$. Let $\mathbf{\varphi}\in\R^{n}$ denote the voltage potentials induced at the vertices in the above setting. By Ohm's law we have $\mathbf{f}=WB\varphi$. Putting these facts together: $$\chi_u-\chi_v=B^\T WB\varphi=L\varphi\text{.}$$
Observe that $(\chi_u-\chi_v) \perp \text{ker}(L)$, and hence $\varphi=L^{+}(\chi_u-\chi_v )$. 

The \textit{effective resistance} between vertices $u$ and $v$ in graph $G$, denoted by $R_{uv}$ is defined as the voltage difference between vertices $u$ and $v$, when a unit of current  is injected into $u$ and is extracted from $v$. Thus we have:
\begin{equation}\label{eq:Eff-Res}
R_{uv}=\mathbf{b}_{uv}^\T L^{+}\mathbf{b}_{uv}.
\end{equation}
We also let $R_{uu}:=0$ for any $u\in V$, for convenience. For any matrix $K\in \R^{n\times n}$, we let $R^K_{uv}:=\mathbf{b}_{uv}^\T K^{+}\mathbf{b}_{uv}$.

Also, for any pair of vertices $(w_1,w_2)$, the potential difference induced on this pair when sending a unit of flow from $u$ to $v$ can be calculated as: 
\begin{equation}
\varphi(w_1)-\varphi(w_2)=\mathbf{b}_{w_1w_2}^\top L^+\mathbf{b}_{uv}.
\end{equation}
Furthermore, if the graph is unweighted, the flow on edge $(w_1,w_2)$ is 
\begin{equation}\label{eq:flow}
\mathbf{f}_{uv}(w_1w_2)=\mathbf{b}_{w_1w_2}^\top L^+\mathbf{b}_{uv}.
\end{equation}

We frequently use the following simple fact.

\begin{fact}[See e.g. \cite{kapralov2017single}, Lemma 3]\label{fact:maxphi2}
	For any graph $G=(V,E,w)$, $\gamma\ge 0$ and any Laplacian matrix $L\in \R^{V}$, let $K= L+\gamma I$. Then, for any pair of vertices $(u,v), (u',v')\in V \times V$, 
	\begin{align*}
	|\mathbf{b}_{u'v'}^\top K^+\mathbf{b}_{uv}|\le \mathbf{b}_{uv}^\top K^+\mathbf{b}_{uv}.
	\end{align*}
\end{fact}
\begin{proof}
	Let $\varphi = K^+\mathbf{b}_{uv}$. Suppose that for some $x\in V\setminus \{u\}$, $\varphi(x)>\varphi(u)$. Then, since $K=L+\gamma I$ is a full rank and diagonally dominant matrix, then one can easily see that we should have $b_{uv}(x)>0$, which is a contradiction. So, $\varphi(u)\ge \varphi(x)$ for any $x \in V\setminus \{u\}$ . In a similar way, we can argue that $\varphi(v)\le \varphi(y)$ for any $y \in V\setminus \{v\}$.  So, the claim holds. 
\end{proof}

\noindent{\textbf{Spectral Approximation}.}
For matrices $C,D \in \R^{p \times p}$, we write $C\preceq D$, if $\forall x \in \R^p$, $x^\T C x \leq x^\T D x$. We say that $\wt{C}$ is $(1\pm\e)$-spectral sparsifier of $C$, and we write it as $\wt{C}\approx_\e C$, if $(1-\e)C \preceq \wt{C} \preceq (1+\e)C$.  Graph $\wt{G}$ is $(1\pm\e)$--spectral sparsifier of graph $G$ if, $L_{\wt{G}}\approx_\e L_G$. We also sometimes use a slightly weaker notation $(1-\e)C \preceq_r \wt{C} \preceq_r (1+\e)C$, to indicate that $(1-\e)x^\top C x\le x^\top \wt{C} x \le (1+\e)x^\top C x$, for any $x$ in the row span of $C$.

\newcommand{\Unif}{\mathrm{Unif}}

\section{Main result}\label{sec:HEanalysis}
We start by giving some intuition and presenting the high level idea of our algorithm in Section~\ref{sec:overview} below. In Section~\ref{sec:main-alg} we formally  state the algorithm and provide correctness analysis. In Section~\ref{sec:sketches} we describe how the required sketches can be implemented using the efficient pseudorandom number generator from~\cite{KMMMN19}. Finally in Section~\ref{sec:proof} we give the proof of Theorem~\ref{thm:main}.

\subsection{Overview of the approach}\label{sec:overview}

To illustrate our approach, suppose for now that our goal is to find edges with effective resistance at least $\frac{1}{\log n}$ in a graph $G=(V,E)$, which we denote by "heavy edges". This task has been studied in prior work on spectral sparsification \cite{kapralov2017single} and was essentially shown in \cite{KMMMN19} to be sufficient to yield a spectral sparsification with only almost constant overhead. Each of \cite{kapralov2017single} and \cite{KMMMN19} solve this problem by running $\ell_2$-heavy hitters on approximate flow vectors, obtained by coarse sparsifier of the graph. The number of test flow vectors used in \cite{kapralov2017single} is quadratic in the number of vertices, i.e., they brute force on all pair of vertices to find the heavy edges, and this was improved to $n^{1.4 + o(1)}$ in \cite{KMMMN19}. Consequently, a natural question that one could attack to further improve the running times of these methods is the following:

\begin{center}
	\fbox{
		\parbox{0.9\textwidth}{ \begin{center}
	Can we efficiently find a nearly linear number of test vectors that enable us to recover all heavy edges?
\end{center}}}
\end{center}

In this work, we answer this question in the affirmative and formally show that there exist a linear number of test vectors, which suffice to find all heavy edges. This is essentially the key technical contribution of this paper and generalizing this solution yields our main algorithmic results.

To illustrate our approach, suppose that one can compute the flow vector using the following formula \footnote{Note that in the actual algorithm we use $K^+$ as opposed to $L^+$, since we work with regularized versions of the Laplacian of $G$, denoted by $K$. We use $L$ in this overview of our techniques to simplify notation.}
\begin{align}
BL^+\mathbf{b}_{uv}=\mathbf{f}_{uv}
\end{align}
for any pair of vertices in polylogarithmic time (in our actual algorithms we will be unable to compute these flow vectors exactly). Note that
\begin{align}
||\mathbf{f}_{uv}||_2^2=\mathbf{b}_{uv}^\top L^+B^\top B L^+\mathbf{b}_{uv}=\mathbf{b}_{uv}^\top L^+\mathbf{b}_{uv}=R_{uv}
\end{align} 
and 
\begin{align}
\mathbf{f}_{uv}(uv)=\mathbf{b}_{uv}^\top L^+ \mathbf{b}_{uv}=R_{uv}.
\end{align}
This implies that, when $R_{uv}>\frac{1}{\log n}$, the contribution of $uv$ coordinate of this vector to the $\ell_2$ norm is substantial, and known $\ell_2$-heavy hitters can recover this edge using corresponding sketches, efficiently. One should note that $\ell_2$-heavy hitter returns a set of edges with $\Omega(\frac{1}{\polylog n})$ contribution to the $\ell_2^2$ of the flow vector. A natural question that arises is whether it is possible to recover a heavy edge without using its flow vector, but rather using other flow vectors. Consider the following example. 
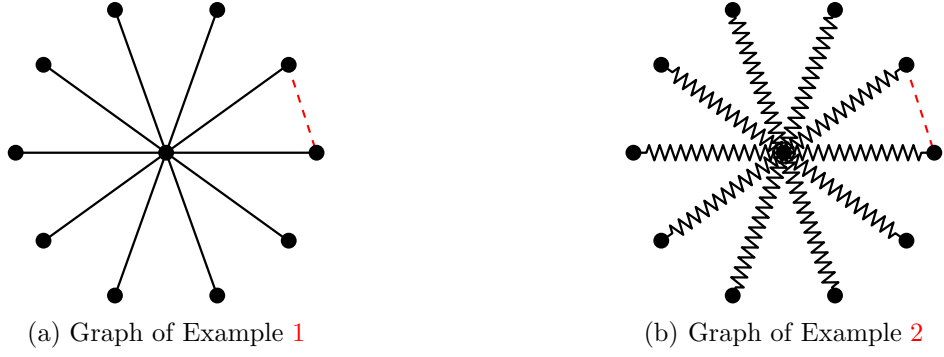
\begin{figure}
	\centering 
	\begin{subfigure}[b]{0.3\textwidth}
		\centering
	\begin{tikzpicture}[line width=0.3mm]

	\draw (0,0)--(+1.62,-1.17);
	\draw (0,0)--(-1.62,-1.17);
	\draw (0,0)--(-1.62,+1.17);

	\draw (0,0)--(+0.68,-1.90);
	\draw (0,0)--(+1.62,1.17);
	\draw (0,0)--(+0.68,1.90);
	\draw (0,0)--(-0.68,1.90);
	\draw (0,0)--(-0.68,-1.90);
	\draw(0,0)--(2,0);
	\draw(0,0)--(-2,0);	

	\draw[red,thick,dashed] (+1.62,+1.17) -- (2,0);
	\fill [color=black] (0,0) circle (3pt);	
	\fill [color=black] (2,0) circle (3pt);
	\fill [color=black] (-2,0) circle (3pt);	
	\fill [color=black] (1.63,-1.17) circle (3pt);
	\fill [color=black] (1.63,+1.17) circle (3pt);
	\fill [color=black] (+0.68,-1.90) circle (3pt);
	\fill [color=black] (+0.68,1.90) circle (3pt);
	\fill [color=black] (-0.68,-1.90) circle (3pt);
	\fill [color=black] (-0.68,1.90) circle (3pt);
	\fill [color=black] (-1.63,-1.17) circle (3pt);
\fill [color=black] (-1.63,+1.17) circle (3pt);
	

	\end{tikzpicture}
	\caption{Graph of Example~\ref{ex:star}}\label{subfig:star1}
	\end{subfigure}
	\hspace{3cm}
		\begin{subfigure}[b]{0.3\textwidth}
		\centering
	\begin{tikzpicture}[line width=0.3mm]

	\draw[red,thick,dashed] (+1.62,+1.17) -- (2,0);
	\fill [color=black] (0,0) circle (3pt);	
	\fill [color=black] (2,0) circle (3pt);
	\fill [color=black] (-2,0) circle (3pt);	
	\fill [color=black] (1.63,-1.17) circle (3pt);
	\fill [color=black] (1.63,+1.17) circle (3pt);
	\fill [color=black] (-1.63,-1.17) circle (3pt);
	\fill [color=black] (-1.63,+1.17) circle (3pt);

	\fill [color=black] (+0.68,-1.90) circle (3pt);
	\fill [color=black] (+0.68,1.90) circle (3pt);
	\fill [color=black] (-0.68,-1.90) circle (3pt);
	\fill [color=black] (-0.68,1.90) circle (3pt);
	\draw [-,
	line join=round,
	decorate, decoration={
		zigzag,
		segment length=4,
		amplitude=2.9,post=lineto,
		post length=3pt
	}] (0,0)--(+1.62,-1.17);
	\draw [-,
	line join=round,
	decorate, decoration={
		zigzag,
		segment length=4,
		amplitude=2.9,post=lineto,
		post length=3pt
	}] (0,0)--(+0.68,-1.90);
	
	\draw [-,
	line join=round,
	decorate, decoration={
		zigzag,
		segment length=4,
		amplitude=2.9,post=lineto,
		post length=3pt
	}] (0,0)--(+1.62,1.17);

	\draw [-,
line join=round,
decorate, decoration={
	zigzag,
	segment length=4,
	amplitude=2.9,post=lineto,
	post length=3pt
}] (0,0)--(-1.62,1.17);

	\draw [-,
line join=round,
decorate, decoration={
	zigzag,
	segment length=4,
	amplitude=2.9,post=lineto,
	post length=3pt
}] (0,0)--(-1.62,-1.17);
	
	\draw [-,
	line join=round,
	decorate, decoration={
		zigzag,
		segment length=4,
		amplitude=2.9,post=lineto,
		post length=3pt
	}] (0,0)--(+0.68,1.90);
	
	\draw [-,
	line join=round,
	decorate, decoration={
		zigzag,
		segment length=4,
		amplitude=2.9,post=lineto,
		post length=3pt
	}] (0,0)--(-0.68,1.90);
	
	\draw [-,
	line join=round,
	decorate, decoration={
		zigzag,
		segment length=4,
		amplitude=2.9,post=lineto,
		post length=3pt
	}] (0,0)--(-0.68,-1.90);
	
	\draw [-,
	line join=round,
	decorate, decoration={
		zigzag,
		segment length=4,
		amplitude=2.9,post=lineto,
		post length=3pt
	}] (0,0)--(2,0);

	\draw [-,
	line join=round,
	decorate, decoration={
		zigzag,
		segment length=4,
		amplitude=2.9,post=lineto,
		post length=3pt
	}] (0,0)--(-2,0);

	\end{tikzpicture}
	\caption{Graph of Example~\ref{ex:star2}}\label{subfig:star2}
	\end{subfigure}
	\caption{(a) \text{Graph of Example~\ref{ex:star}}. A star with $n$ petals along with one additional edge. (b) Graph of Example~\ref{ex:star2}. A star graph with $\Theta(n^{0.7})$ petals, along with one additional edge. Each zigzag represents a path of connected cliques with effective resistance diameter $O(1)$.}\label{fig:star}
\end{figure}
\begin{example}[Star Graph Plus Edge]\label{ex:star}
	Suppose that graph $G=(V,E)$ is a ``star" with a center and $n$ petals along with one additional edge that connects a pair of petals, i.e., $V=\{v_1,v_2,\dots,v_n\}$ and $E=\{(v_1,v_2),(v_1,v_3),\cdot,(v_1,v_n)\}\cup\{(v_2,v_3)\}$ (see Figure~\ref{subfig:star1}).
	
	 Clearly, for edge $(v_2,v_3)$, $R_{v_2v_3} = \frac{2}{3}$. Suppose that we want to recover this edge by  examining an electrical flow vector other than $\mathbf{f}_{v_2v_3}$. We can in fact pick an arbitrary vertex $x\in V \setminus v_2 $ and send one unit of flow to $v_2$. Regardless of the choice of $x$, edge $(v_2,v_3)$ contributes an $\Omega(1)$ fraction of the energy of the flow, and thus can be recovered by applying heavy hitters to $\mathbf{f}_{xv_2}$. Similarly, for any $v_i$, when one unit of flow is sent from $x$ to $v_i$, at least a constant fraction of the energy is contributed by edge $(x,v_i)$. So, all high effective resistance edges in this graph (all edges) can be recovered using $n-1$ simple flow vectors, i.e., $\{\mathbf{f}_{xv_1},\dots ,\mathbf{f}_{xv_n}\}$.
\end{example}

Of course, the graph in Example~\ref{ex:star} has only $n$ edges, and so could be stored explicitly in the streaming setting, without needing to recover edges from heavy hitter queries. However, we can give a similar example which is in fact dense.


\begin{example}[Thick Star Plus Edge] \label{ex:star2}
	Suppose that graph $G$ is a dense version of the previous example as follows: it has a center and $\Theta(n^{0.7})$ petals. Each petal consists of a chain of $\Theta(n^{0.2})$ cliques of size $n^{0.1}$, where each pair of consecutive cliques is connected with a complete bipartite graph. One can verify that the effective resistance diameter of each petal is $\Theta(1)$. Now, we add an additional edge, $e$, that connects an arbitrary node in the leaf of one petal to a node in the leaf of another petal (see Figure~\ref{subfig:star2}). 
	
	As in Example \ref{ex:star}, $e$ is heavy, with $R_e = \Theta(1)$. In fact, it is the only heavy edge in the graph.
	One can verify that, similar to Example~\ref{ex:star}, if we let $C_2$ and $C_3$ denote the cliques that $e$ connects, choosing an arbitrary vertex $x$ and sending flow to any node in $C_2$ and then to any node in $C_3$, will give an electrical flow vector where $e$ contributes an $\Omega(1)$ fraction of the energy. Thus, $e$ can be recovered by applying heavy hitters to these vectors. Consequently, using $n$ test vectors (sending flow from $x$ to each other node in the graph) one can recover all heavy edges of this example. 
\end{example}

Unfortunately, it is possible  to give an example where the above simple procedure of checking the flow from an arbitrary vertex to all others fails.
\begin{example}[Thick Line Plus Edge] \label{ex:1}
	Suppose that graph $G=(V,E)$ is a thick line, consisting of $n^{0.9}$ set of points (clusters) where any two consecutive clusters form a complete bipartite graph. Formally,  $V=\{v_1,v_2,\dots,v_n\}=C_1\cup C_2 \cup \dots \cup C_{n^{0.9}}$, where $C_i$'s are disjoint sets of size $n^{0.1}$ and $$E= \bigcup_{i=1}^{n^{0.9}-1} C_i \times C_{i+1}. $$ Also, add an edge $e=(u,v)$ such that $u\in C_1$ and $v\in C_{n^{0.2}}$ (see Figure~\ref{fig:line2}). 
	
	One can verify that $R_e = \Omega(1)$. However,
	if one picks an arbitrary vertex $x\in V$ and sends one unit of flow each other vertex, running $\ell_2$-heavy hitters on each of these flows will not recover edge $e$ if $x$ is far from $u$ and $v$ in the thick path. Any flow that must cross $(u,v)$ will have very large energy due to the fact that it must travel a long distance to the clusters containing these vertices, so $e$ will not contribute non-trivial fraction. 
\end{example}
\begin{figure}
	\centering
	\begin{tikzpicture}[line width=0.3mm]
	
	\fill [color=black] (3.25,0) circle (1pt);
	\fill [color=black] (3.5,0) circle (1pt);
	\fill [color=black] (3.75,0) circle (1pt);	
	
	\fill [color=black] (0,0) circle (3pt);	
	\fill [color=black] (1,0) circle (3pt);
	\fill [color=black] (2,0) circle (3pt);
	\fill [color=black] (3,0) circle (3pt);
	\fill [color=black] (4,0) circle (3pt);
	\fill [color=black] (5,0) circle (3pt);
	\fill [color=black] (6,0) circle (3pt);
	
	\fill [color=black] (6.25,0) circle (1pt);
	\fill [color=black] (6.5,0) circle (1pt);
	\fill [color=black] (6.75,0) circle (1pt);

	\fill [color=black] (7,0) circle (3pt);
	\fill [color=black] (8,0) circle (3pt);
	\fill [color=black] (9,0) circle (3pt);
	\fill [color=black] (10,0) circle (3pt);
	
	\draw (0,-0.75) node[anchor=south] {$C_1$} ;
	\draw (1,-0.75) node[anchor=south] {$C_2$} ;
	\draw (2,-0.75) node[anchor=south] {$C_3$} ;
	\draw (3,-0.75) node[anchor=south] {$C_4$} ;
	\draw (5,-0.75) node[anchor=south] {$C_{n^{0.2}}$} ;
	\draw (10,-0.75) node[anchor=south] {$C_{n^{0.9}}$} ;
	
	\draw [-,
	line join=round,
	decorate, decoration={
		zigzag,
		segment length=4,
		amplitude=2.9,post=lineto,
		post length=3pt
	}] (0,0)--(1,0);
	\draw [-,
	line join=round,
	decorate, decoration={
		zigzag,
		segment length=4,
		amplitude=2.9,post=lineto,
		post length=3pt
	}] (1,0)--(2,0);
	\draw [-,
	line join=round,
	decorate, decoration={
		zigzag,
		segment length=4,
		amplitude=2.9,post=lineto,
		post length=3pt
	}] (2,0)--(3,0);
	\draw [-,
	line join=round,
	decorate, decoration={
		zigzag,
		segment length=4,
		amplitude=2.9,post=lineto,
		post length=3pt
	}] (4,0)--(5,0);
	\draw [-,
	line join=round,
	decorate, decoration={
		zigzag,
		segment length=4,
		amplitude=2.9,post=lineto,
		post length=3pt
	}] (5,0)--(6,0);

	\draw [-,
	line join=round,
	decorate, decoration={
		zigzag,
		segment length=4,
		amplitude=2.9,post=lineto,
		post length=3pt
	}] (7,0)--(8,0);
	\draw [-,
	line join=round,
	decorate, decoration={
		zigzag,
		segment length=4,
		amplitude=2.9,post=lineto,
		post length=3pt
	}] (8,0)--(9,0);
	\draw [-,
	line join=round,
	decorate, decoration={
		zigzag,
		segment length=4,
		amplitude=2.9,post=lineto,
		post length=3pt
	}] (9,0)--(10,0);
	\draw [red,dashed]   (0,0) to[out=+90,in=+90, distance=2cm ] (5,0);
	\end{tikzpicture}
	\caption{\text{Graph of Example~\ref{ex:1}}. Each $C_i$ represents a cluster with $n^{0.1}$ vertices (with no internal edges) and each zigzag represents the edges of a complete bipartite graph between consecutive $C_i$'s.}\label{fig:line2}
\end{figure}
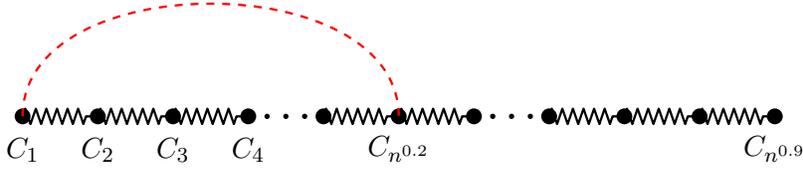

Fortunately, the failure of our recovery method in Example~\ref{ex:1} is due to a simple fact: the effective resistance diameter of the graph is large. When the effective resistance diameter is small (as in Examples~\ref{ex:star} and ~\ref{ex:star2}) the strategy always suffices. This follows from the following simple observation:
\begin{observation}\label{obs:obs1}
	For a graph $G=(V,E)$, suppose that for an edge $e = (u,v)\in E$, one has $$R_{e}\ge \beta.$$ Then, for any $x \in V$, in at least one of these settings, edge $e$ carries at least $\beta/2$ units of flow:
	\begin{enumerate}
		\item {One unit of flow is sent from $x$ to $u$.}
		\item {One unit of flow is sent from $x$ to $v$.}
	\end{enumerate}
\end{observation}
This observation follows formally from the following simple lemma.
\begin{lemma}{\label{lem:halfeff}}
	For a graph $G=(V,E)$, suppose that $D\in \R^{V\times V}$ is a PSD matrix. Then, for any pair of vertices $(u,v)\in {V\choose2}$ and for any vertex $x \in V\setminus \{u,v\}$, 
	\begin{align*}
	\max\{|\mathbf{b}_{xu}^\top D \mathbf{b}_{uv}|,|\mathbf{b}_{xv}^\top D \mathbf{b}_{uv}|\} \ge \frac{\mathbf{b}_{uv}^\top D \mathbf{b}_{uv}}{2}.
	\end{align*}
\end{lemma}
\begin{proof}
	Note that 
	\begin{align*}
	\mathbf{b}_{uv}^\top D \mathbf{b}_{uv}&=\left(\mathbf{b}_{ux}^\top + \mathbf{b}_{xv}^\top \right) D \mathbf{b}_{uv}
	=\mathbf{b}_{ux}^\top D \mathbf{b}_{uv}+\mathbf{b}_{xv}^\top D \mathbf{b}_{uv},
	\end{align*}
	and hence, the claim holds. 
\end{proof}
Consider the setting where $\beta = \frac{1}{\log n}$. The observation guarantees that edge $e$ contributes at least $\frac{1}{4\log^2 n}$ energy to either flow $\mathbf{f}_{xv}$ or $\mathbf{f}_{xu}$. Thus, we can recover this edge via $\ell_2$ heavy hitters, as long as the total energy $\norm{\mathbf{f}_{xv}}_2^2$ or $\norm{\mathbf{f}_{xu}}_2^2$ is not too large. Note that this energy is just equal to the effective resistance $R_{xv}$ between $x$ and $v$ (respectively $x$ and $u$). Thus it is bounded if the effective resistance diameter is small, demonstrating that our simple recovery procedure always succeeds in this setting. For example, if the diameter is $O(1)$, both $\norm{\mathbf{f}_{xv}}_2^2=O(1)$ and $\norm{\mathbf{f}_{xu}}_2^2 = O(1)$, and so by Observation \ref{obs:obs1}, edge $e = (u,v)$ contributes at least a $\Theta \left ( \frac{1}{\log^2 n} \right )$ fraction of the energy of at least one these flows.

We next explain how to extend this procedure to handle general graphs, like that of Example \ref{ex:1}.

\paragraph{Ball carving in effective resistance metric:} When the effective resistance diameter of $G$ is large, if we attempt to recover $e$ using $\ell_2$-heavy hitters on the flow vectors $\mathbf{f}_{xu}$ and $\mathbf{f}_{xv}$, for an arbitrary chosen $x\in V$, we may fail if the effective resistance distance between $x$ and $v$ or $u$ ($\norm{\mathbf{f}_{xv}}_2^2$ or $\norm{\mathbf{f}_{xu}}_2^2$) is large. This is exactly what we saw in Example \ref{ex:1}.

However, using the fact that $||\mathbf{f}_{xv}||_2^2=R_{xv}$, our test will succeed if we find a vertex $x$, which is close to $u$ and $v$ in the effective resistance metric. This suggests that we should partition the vertices into cells of fairly small effective resistance diameter, ensuring that both endpoints of an edge $(u, v)$ that we would like to recover fall in the same cell with nontrivially large probability. This is exactly what standard metric decomposition techniques achieve through a ball-carving approach, which we use, as described next.

\paragraph{Partitioning the graph into low effective resistance diameter sets:} It is well-known that using Johnson-Lindenstrauss (JL) dimension reduction (see Lemma~\ref{lem:JL}), one can embed vertices of a graph in $\R^q$, for $q=O(\log n)$, such that the Euclidean distance squares correspond to a constant factor multiplicative approximation to effective resistance of corresponding vertices. We then partition $\R^q$  intro $\ell_\infty$ balls centered at points of a randomly shifted infinite $q$-dimensional grid with side length $w>0$, essentially defining a hash function that maps every point in $\R^q$ to the nearest point on the randomly shifted grid.  We then bound the maximum effective resistance of pair of vertices in the same bucket (see Claim~\ref{lem:celldiambound}), and show how an appropriate choice of the width $w$ ensures that $u$ and $v$ belong to the same cell, with a probability no less than a universal constant (see Claim~\ref{fact:constantedge}). This ensures that in at least one of $O(\log n)$ independent repetitions of this process with high probability, $u$ and $v$ fall into the same cell.  We note that the parameters of our partitioning scheme can be improved somewhat using Locality Sensitive Hashing techniques (e.g.,~\cite{IndykM98,DatarIIM04,AndoniI06,AndoniINR14,AndoniR15}). More precisely, LSH techniques would improve the space complexity by polylogarithmic factors at the expense of slightly higher runtime (the best improvement in space complexity would result from Euclidean LSH~\cite{AndoniI06,AndoniR15}, at the cost of an $n^{o(1)}$ additional factor in runtime). However, since the resulting space complexity does not quite match the lower bound of $\Omega(n\log^3 n)$ due to ~\cite{NelsonY19}, we leave the problem of fine-tuning the parameters of the space partitioning scheme as an exciting direction for further work.

\paragraph{Sampling edges with probability proportional to effective resistances:} The above techniques can actually be extended to recover edges of any specific target effective resistance. Broadly speaking, if we aim to capture edges  of effective resistance about $R$, we can afford to lower our grid cell size proportionally to $R$. Unfortunately, these edges don't contribute enough to the flow vector to be recoverable. Thus, we will also subsample the edges of the graph at rate approximately proportional to $R$ to allow us to detect the target edges while also subsampling them.

\subsection{Our algorithm and proof of main result}\label{sec:main-alg}

As mentioned in the introduction, our algorithm consists of two phases. In the first phase, our algorithm maintains sketches of the stream, updating the sketches at each edge addition or deletion. Then, in the second phase, when queried, it can recover a spectral sparsifier of the graph from the sketches that have been maintained in the first phase. In the following lines, we give a brief overview of each phase:
\paragraph{Updating sketches in the dynamic stream.} Our algorithm maintains a set of sketches $\Pi B$, of size $O(n\polylog(n)\cdot\e^{-2})$, and updates them each time it receives an edge addition or deletion in the stream. $\Pi B$ consists of multiple sketches $(\Pi_s^\ell B_s^\ell)_{\ell,s}$ where $B_s^\ell$ is a subsampling of the edges in $B$ at rate $2^{-s}$ and $\Pi^\ell_s$ is an $\ell_2$ heavy hitters sketch. In section~\ref{sec:sketches} we discuss these sketches in more detail and we show that the update time for each edge addition or deletion is $O(\polylog(n)\cdot \e^{-2})$. 
\paragraph{Recursive sparsification:}
After receiving the updates in the form of a dynamic stream, as described above, our algorithm uses the maintained sketches to recover a spectral sparsifier of the graph. This is done recursively, and heavily relies on the idea of a chain of coarse sparsifiers described in Lemma~\ref{lem:chain_coarse}. For a regularization parameter $\ell$ between $0$ and $d = O(\log n)$ the task of \textsc{Sparsify}($\Pi^{\le\ell} B,\ell,\epsilon$) is to output a spectral sparsifier to matrix $L_\ell$, which is defined as follows:
	\begin{align*}
L_{\ell}=\begin{cases}
L_G+\frac{\lambda_u}{2^\ell}I & \text{if $0\le\ell\le d$}\\
L_G & \text{if $\ell=d+1.$}
\end{cases}
\end{align*}
where $d = \ceil{\log_2 \frac{\lambda_u}{\lambda_\ell}}$ (see Lemma~\ref{lem:chain_coarse} for more details about chain of coarse sparsifiers). Note that the call receives a collection of sketches $\Pi^{\le \ell} B$ as input that suffices for all recursive calls with smaller values of $\ell$. So, in order to get a sparsifier of the graph we invoke \textsc{Sparsify}($\Pi^{\le d+1} B,d+1,\epsilon$), which receives all the sketches maintained throughout the stream and passes the required sketches to the recursive calls in line~\ref{line:HE-CallSp} of Algorithm~\ref{alg:main-sparsify}. This recursive algorithm takes as input $\Pi^{\le\ell}B$ corresponding to the parts of the sketch used to recover a spectral approximation to $L_{k}$ for all $k \leq \ell$, $\ell$ corresponding to the current $L_\ell$ which we wish to recover a sparsifier of, and $\epsilon$ corresponding to the desired sparsification accuracy. The algorithm first invokes itself recursively to recover $\tilde{K}$, a spectral approximation for $L_{\ell - 1}$ (or uses the trivial  approximation $\lambda_u I$ when $\ell = 0$). The effective resistance metric induced by $\tilde{K}$ is then approximated using the Johnson-Lindenstrauss lemma (JL). Finally, the procedure  \textsc{RecoverEdges} (i.e. Algorithm~\ref{alg:main-heavy-edge}) uses this metric and the heavy hitters sketches $(\Pi^\ell_s B_s^\ell)_s$. We formally state our algorithm, Algorithm~\ref{alg:main-heavy-edge} below.
\begin{algorithm}[H]
	\caption{\textsc{Sparsify}($\Pi^{\le\ell} B,\ell,\epsilon$)}  
	\label{alg:main-sparsify} 
	\begin{algorithmic}[1]
		
		\Procedure{Sparsify($\Pi^{\le\ell}B,\ell,\epsilon$)}{}
			\State $W \gets 0^{n\times n}$
			\If{$\ell= 0$} 
				\State $\wt{K}\gets \lambda_u I$\label{line:leaf-case-K}
			\Else
				\State $\wt{K} \gets \frac{1}{2(1+\e)}\textsc{Sparsify}(\Pi^{\le\ell-1}B,\ell-1,\epsilon)$ \label{line:HE-CallSp} 
			\EndIf		
			\State $\wt{B} \gets$ \text{the edge vertex incident matrix of $\wt{K}$ (discarding the regularization)}
			\State $\wt{W} \gets$ \text{the diagonal matrix of weights over the edges of  $\wt{K}$ (discarding the regularization)} 
			\State $Q\gets q\times {n \choose 2}$ is a random $\pm$1 matrix\label{line:setQ} for $q\gets 1000 \log n$ 
			\State \Comment{\emph{$q$ above is chosen as it suffices to get a $(1\pm \frac{1}{5})$ approximation from JL}}  \label{line:setq}  
			\State  $M\gets \frac1{\sqrt{q}}Q\wt W^{1/2} \wt B\wt K^+$ \Comment{\emph{$M$ is such that  $R_{uv}^{\wt{K}}\le\frac{5}{4}||M(\chi_u-\chi_v)||_2^2\le \frac{3}{2} R_{uv}^{\wt{K}}$ w.h.p.}} \label{line:setM}
			\State \Comment{\emph{$R^{\wt{K}}$ is the effective resistance metric in $\wt{K}$}}
			\For {$s\in [-\log\left(3\cdot c_2\cdot\log n\cdot\e^{-2}\right),10\log n]$} \label{line:s-for} 
				\State $E_s\gets\textsc{RecoverEdges}(\Pi^\ell_{s^+}B^\ell_{s^+},M,\wt K^+,s,q,\e)$\label{line:recover} \Comment{\emph{We use the notation $s^+=\text{max}(0, s)$}}
				\For{$e\in E_s$} \State $W(e,e)\gets2^{-(s^+)}$ \EndFor
			\EndFor
			\If {$\ell= \left\lceil \log_2 \frac{\lambda_u}{\lambda_\ell}\right\rceil+1$}
				\State $\gamma\gets 0$
			\Else
				\State $\gamma \gets \frac{\lambda_u}{2^\ell}$
			\EndIf
			\State \Return $B_n^\top W B_n+ \gamma I$.
		\EndProcedure
	\end{algorithmic}
\end{algorithm}

Algorithm~\ref{alg:main-heavy-edge} (the \textsc{RecoverEdges} primitive) is the core of Algorithm~\ref{alg:main-sparsify}. It receives a parameter $s$ as input, and its task is to recover edge of effective resistance $\approx \frac{\e^2}{\log n} 2^{-s}$ from a sample at rate $\min(1, O(\frac1{\e^2}\log n\cdot 2^{-s}))$ from an appropriate sketch. It is convenient to let $s$ range from $-O(\log (\log n/\e^2))$ to $O(\log n)$, so that the smallest value of $s$ corresponds to edges of constant effective resistance. That way the sampling level corresponding to $s$ is simply equal to $s^+:=\max(0, s)$. Therefore Algorithm~\ref{alg:main-heavy-edge} takes as input  a heavy hitters sketch $\Pi_{s^+}^\ell B_{s^+}^\ell$ of $B_{s^+}^\ell$, the edge incidence matrix of $L_\ell$  sampled at rate $2^{- s^+}$, an approximate effective resistance embedding $M$, the target sampling probability $2^{-s}$, the dimension $q$ of the embedding, and the target accuracy $\epsilon$. This procedure then performs the previously described random grid hashing of the points using the effective resistance embedding and queries the heavy hitters sketch to find the edges sampled at the appropriate rate.

The development and analysis of \textsc{RecoverEdges} (Algorithm~\ref{alg:main-heavy-edge})  is  the main technical contribution of our paper.  In the rest of the section we prove correctness of Algorithm~\ref{alg:main-heavy-edge} (Lemma~\ref{lm:main}, our main technical lemma), and then provide a correctness proof for Algorithm~\ref{alg:main-sparsify}, establishing Theorem~\ref{lem:corr-main-HE}. We then put these results together with runtime and space complexity bounds to obtain a proof of Theorem~\ref{thm:main}.

\begin{algorithm}[h!]
	\caption{\textsc{RecoverEdges}($\Pi^\ell_{s^+}B^\ell_{s^+},M,\wt K^+,s,q,\e$)}  
	\label{alg:main-heavy-edge} 
	\begin{algorithmic}[1]
		\Procedure{RecoverEdges($\Pi^\ell_{s^+}B^\ell_{s^+},M,\wt K^+,s,q,\e$)}{}
			\State $E'\gets\emptyset$.\Comment{\emph{$q$ is the dimension to perform hashing in, $s$ is the sampling level}}
			\State $C \gets$ the constant in the proof of Lemma~\ref{lm:main}
			\State $c_2 \gets$ the oversampling constant of Theorem~\ref{lem:classic_result}
			\State $w\gets2q\cdot\sqrt{\frac{\e^2}{c_2\cdot2^s\cdot\log n}}$.\label{line:w-set}
			\For {$j\in\left[10\log n\right]$}\label{line:j-for}
				\State For each dimension $i\in [q]$, choose $s_i \sim \Unif([0,w])$.
				\State Initialize $H\gets \emptyset$ to an empty hash table
				\For {$u\in V$} \Comment{\emph{Hash vertices to points on randomly shifted grid}}
				\State For all $i\in [q]$, let $\mathcal{G}(u)_i:=\left\lfloor \frac{(M\chi_u)_i-s_i}{w}\right\rfloor$. \label{line:set-hashing}
				\State Insert $u$ into $H$ with key $\mathcal{G}(u)\in \mathbb{Z}^q$ \label{line:set-hashing-hash}
				\State \Comment{\emph{$\mathcal{G}(u)\in \mathbb{Z}^q$ indexes a  point on a randomly shifted grid}}
				\EndFor
				\For {$b\in \text{keys}(H)$}\Comment{\emph{$b\in \mathbb{Z}^q$ indexes a point on a randomly shifted grid}}
					\State $x\gets $arbitrary vertex in $H^{-1}(b)$\label{line:x}
					\For {$v\in H^{-1}(b)\setminus \{x\}$}
					\State  $F \gets \textsc{HeavyHitter}\left(\Pi^\ell_{s^+}B^\ell_{s^+}\wt{K}^+\mathbf{b}_{xv},  \frac{1}{2} \cdot \frac{1}{C\cdot q^3}\cdot \sqrt{\frac{\epsilon^2}{ \log n}}\right)$. \Comment{\emph{As per Lemma~\ref{lem:HH}}} \label{line:heavy-hitter}
					\For {$e \in F$ }\label{line:e-for}
						\State $p_e'\gets\frac54\cdot 	c_2\cdot||Mb_e||_2^2\cdot\log n/\e^2$ \label{def:p_e} 
						\If{$p_e'\in(2^{-s-1},2^{-s}]}$\label{p_e-condition}\label{line:verify}
							\State $E'\gets E'\cup\{e\}$.
						\EndIf
					\EndFor
				\EndFor
				\EndFor
			\EndFor
			\State \Return $E'$.
		\EndProcedure
	\end{algorithmic}
\end{algorithm}

Lemma~\ref{lm:main} below is our main technical lemma. Specifically, Lemma~\ref{lm:main} proves that if  Algorithm~\ref{alg:main-sparsify} successfully  executes all lines before line~\ref{line:s-for}, then each edge is sampled and weighted properly (as required by Theorem~\ref{lem:classic_result}), in the remaining steps. 
\begin{lemma}[Edge Recovery] \label{lm:main}
	Consider an invocation of $\textsc{RecoverEdges}(\Pi^\ell_{s^+} B^\ell_{s^+},M,\wt K,s,q,\e)$ of Algorithm~\ref{alg:main-heavy-edge}, where $\Pi^\ell_{s^+} B^\ell_{s^+}$ is a sketch of the edge incidence matrix $B$ of the input graph $G$ as described in Section~\ref{sec:sketches}, $s$ is some integer, and $\e \in (0,1/5)$. Suppose further that $\wt K$ and $M$ satisfy the following guarantees:
	\begin{description}
		\item[(A)] $\wt{K}$ is such that $ \frac13 \cdot L_\ell \preceq_r \wt{K} \preceq_r  L_\ell$ (see lines~\ref{line:leaf-case-K} and \ref{line:HE-CallSp} of Algorithm~\ref{alg:main-sparsify})
		\item[(B)] $M$ is such that for any pair of vertices $u$ and $v$,  $R_{uv}^{\wt{K}}\le\frac{5}{4}||M(\chi_u-\chi_v)||_2^2\le \frac{3}{2} R_{uv}^{\wt{K}}$ ($R^{\wt{K}}$ is the effective resistance metric in $\wt{K}$; see line~\ref{line:setM} of Algorithm~\ref{alg:main-sparsify})
	\end{description}
Then, with high probability, for every edge $e$, $\textsc{RecoverEdges}(\Pi^\ell_{s^+} B^\ell_{s^+},M,\wt K,s,q,\e)$ will recover $e$ if and only if:
\begin{description}
	\item[(1)] $\frac54\cdot c_2\cdot||Mb_e||_2^2\cdot\log(n)/\e^2\in (2^{-s-1},2^{-s}]$ where $c_2$ is the oversampling constant of Theorem~\ref{lem:classic_result} (see lines~\ref{def:p_e} and \ref{p_e-condition} of Algorithm~\ref{alg:main-heavy-edge}), and
	\item[(2)] edge $e$ is sampled in $B^\ell_{s^+}$.
\end{description}
\end{lemma}

The proof of Lemma~\ref{lm:main} relies on the following two claims regarding the hashing scheme of Algorithm~\ref{alg:main-heavy-edge}. First, Claim~\ref{fact:constantedge} shows that the endpoints of an edge of effective resistance bounded by a threshold most likely get mapped to the same grid point in the random hashing step in line~\ref{line:set-hashing} of Algorithm~\ref{alg:main-heavy-edge}.

\begin{claim}[Hash Collision Probability]\label{fact:constantedge}
	Let $q$ be a positive integer and let the function $\mathcal{G}: \R^q \rightarrow \mathbb{Z}^q$ define a hashing with width $w>0$ as follows:
	\begin{align*}
	\forall i \in [q], \ \ \mathcal{G}(u)_i=\left\lfloor \frac{u_i-s_i}{w}\right\rfloor
	\end{align*}
	where $s_i \sim \Unif[0,w]$, as per line~\ref{line:set-hashing} of Algorithm~\ref{alg:main-heavy-edge}. If for a pair of points $x,y\in \R^q$, $||x-y||_2\le w_0$ and $w\ge2w_0q$, then $\mathcal{G}(x)=\mathcal{G}(y)$ with probability at least $1/2$.
\end{claim}
\begin{proof}
	First note that by union bound
\begin{equation}\label{eq:19hg9geweg}
\mathbb P(\mathcal G(x)\neq\mathcal G(y)) =\mathbb P(\exists i:\mathcal G(x)_i\neq\mathcal G(y)_i) \le\sum_{i=1}^q\mathbb P(G(x)_i\neq G(y)_i) ~.
\end{equation}
	Now let us bound each term of the sum.
	\begin{equation}\label{eq:239hg23g3g}
	\begin{split}
		\mathbb P(\mathcal G(x)_i\neq\mathcal G(y)_i)&=\mathbb P\left(\left\lfloor\frac{x_i-s_i}{w}\right\rfloor\neq\left\lfloor\frac{y_i-s_i}{w}\right\rfloor\right)\\
		&=\frac{|x_i-y_i|}{w}\\
		&\le\frac{||x-y||_2}{w}\\
		&\le\frac1{2q}
	\end{split}
	\end{equation}
	Combining~\eqref{eq:19hg9geweg} and~\eqref{eq:239hg23g3g}, we get that $\mathbb P(\mathcal G(x)\neq\mathcal G(y))\leq 1/2$ as claimed.
\end{proof}

The next claim, Claim~\ref{lem:celldiambound} bounds the effective resistance diameter of buckets in the hash table constructed in line~\ref{line:set-hashing-hash} of Algorithm~\ref{alg:main-heavy-edge}.

\begin{claim}[Hash Bucket Diameter]{\label{lem:celldiambound}}
	Let the function $\mathcal{G}: \R^q \rightarrow \mathbb{Z}^q$, for some integer $q$, define a hashing with width $w>0$ as follows:
	\begin{align*}
	\forall i \in [q], \ \ \mathcal{G}(u)_i=\left\lfloor \frac{u_i-s_i}{w}\right\rfloor
	\end{align*}
	where $s_i \sim \Unif[0,w]$, as per line~\ref{line:set-hashing} of Algorithm~\ref{alg:main-heavy-edge}.
	For any pair of points $u,v \in \R^q$, such that $\mathcal{G}(u)=\mathcal{G}(v)$, one has $$||u-v||_2 \le w\cdot \sqrt{q}.$$
\end{claim}
\begin{proof}
	Since $\mathcal{G}(u)=\mathcal{G}(v)$, then 
	\begin{align*}
	||u-v||_2&=\sqrt{\sum_{i\in q} (u_i-v_i)^2}\\
	&\le \sqrt{w^2 \cdot q}&&\text{since }\forall i \in q, |u_i-v_i|\le w\\
	&=w\cdot \sqrt{q}.
	\end{align*}
\end{proof}

Using Claim~\ref{fact:constantedge} and Claim~\ref{lem:celldiambound} we now prove Lemma~\ref{lm:main}.

\begin{proofof}{Lemma~\ref{lm:main}}
	 Let $p_e':=\frac54\cdot c_2\cdot||Mb_e||_2^2\cdot\log(n)/\e^2$. First note that both of conditions {\bf (1)} and {\bf (2)} are necessary. Indeed, if $e$ is not sampled in $B_s$, it will never be returned by \textsc{HeavyHitter} in line~\ref{line:heavy-hitter}, and if $p'_e\not\in(2^{-s-1},2^{-s}]$ then $e$ will not be added to $E'$ due to line~\ref{line:verify} of Algorithm~\ref{alg:main-heavy-edge}. It remains to show that the two conditions are sufficient to recover $e$ with high probability.

	For an edge $(u,v)=e\in E$ satisfying conditions {\bf (1)} and {\bf (2)} we prove that the size of the grid ($w$ as defined in line~\ref{line:w-set} of Algorithm~\ref{alg:main-heavy-edge})  is large enough to capture edge $e$, as described by Claim~\ref{fact:constantedge}. Specifically, we invoke the claim with $w_0=||Mb_e||_2$. Note that we have $w\ge2qw_0$ by the setting of $w$ in line~\ref{line:w-set} and the fact that  
\begin{align*}
	||Mb_e||_2 =\sqrt{\frac45\cdot p_e'\cdot\frac{\e^2}{c_2\cdot\log n}}
	\le\sqrt{\frac{\e^2}{c_2\cdot2^s\cdot\log n}},
\end{align*}
where we used the fact that $p_e'\leq 2^{-s}$. Thus, we have $w\ge2qw_0$ as prescribed by  Claim~\ref{fact:constantedge}, so $u$ and $v$ fall into the same cell with probability at least $1/2$ in a single instance of hashing. Hashing is then repeated $10\log n$ times to guarantee that they fall into the same cell at least once with high probability, see line~\ref{line:j-for} of Algorithm~\ref{alg:main-heavy-edge}.

	Consider now an instance of hashing where $u$ and $v$ fall into the same cell, say $\mathcal C$ (which corresponds to a hash bucket in our hash table $H$). Let $x$ be chosen arbitrarily from $\mathcal{C}$ as per line~\ref{line:x} of Algorithm~\ref{alg:main-heavy-edge} . Our algorithm sends electrical flow from $x$ to both $u$ and $v$ and by Observation~\ref{obs:obs1} in at least one of these flows $e$ will have weight $R_e^{\wt K}/2$. More precisely, by Lemma~\ref{lem:halfeff} invoked with $D=\wt{K}^+$ we have 
	\begin{equation}\label{eq:i23g93hg}
	\max\{|\mathbf{b}_{xu}^\top \wt{K}^+ \mathbf{b}_{uv}|,|\mathbf{b}_{xv}^\top \wt{K}^+ \mathbf{b}_{uv}|\} \ge \frac{\mathbf{b}_{uv}^\top \wt{K}^+ \mathbf{b}_{uv}}{2}=R_e^{\wt K}/2.
	\end{equation}

	 Without loss of generality assume that this is the flow from $x$ to $v$.


	It remains to show, that unlike in Example~\ref{ex:1}, the total energy of the $xv$ flow does not overshadow the contribution of edge $e$. Intuitively this is because the effective resistence of $e$ is proportional to $2^{-s}\cdot\e^2$ and therefore its $\ell_2$-contribution is proportional to $2^{-2s}\cdot\e^4$. On the other hand, the effective resistence diameter of $\mathcal C$ is proportional to $2^{-s}\cdot\e^2$, which bounds the energy of the $xv$ flow before subsampling. Subsampling at rate $2^{-s}$ decreases the energy by a factor of $2^{-s}$ in expectation, and the energy concentrates sufficiently around its expectation with high probability. We prove everything in more detail below. It turns out that the actual ratio between contribution of $e$ and the entire energy of the subsampled flow is polylogarithmic in $n$ and quadratic in $\e$. Therefore, we can afford to store a heavy hitter sketch powerful enough to recover $e$.
	
	Now let $\wt{\mathbf{f}}_{xv}=B  \wt{K}^+(\chi_{x}-\chi_{v})$, and $\wt{\mathbf{f}}_{xu}=B  \wt{K}^+(\chi_{x}-\chi_{u})$. Note that $\mathbf{f}_{xv} \in \R^{{n}\choose{2}}$ is a vector whose nonzero entries are exactly the voltage differences across edge in $G$ when one unit of current is forced from $x$ to $v$ in $\wt{K}$. We have, writing $L$ instead of $L_\ell$ to simplify notation,
	\begin{align*}
	||\wt{\mathbf{f}}_{xv}||_2^2&=(\chi_{x}-\chi_{v})^\top\wt{K}^+ B^\top  B  \wt{K}^+(\chi_{x}-\chi_{v}) \\
	&\le (\chi_{x}-\chi_{v})^\top\wt{K}^+ L \wt{K}^+(\chi_{x}-\chi_{v}) && \text{Since }  B^\top  B  \preceq L\\
	&\leq 3\cdot (\chi_{x}-\chi_{v})^\top\wt{K}^+ \wt{K} \wt{K}^+(\chi_{x}-\chi_{v}) && \text{Since $L \preceq 3\cdot \wt{K}$ by assumption {\bf (A)} } \\
	&= 3\cdot (\chi_{x}-\chi_{v})^\top\wt{K}^+ (\chi_{x}-\chi_{v}) &&\text{~~~~~~~~~~~~~~~~~~~~~~of the lemma}\\
	&= 3\cdot R^{\wt{K}}_{xv}
	\end{align*}
	Moreover we have
	$$\wt{\mathbf{f}}_{xv}(uv)=(\chi_{u}-\chi_{v})^\T \wt{K}^{+} (\chi_{x}-\chi_{v})$$
	and
	$$\wt{\mathbf{f}}_{xu}(uv)=(\chi_{u}-\chi_{v})^\T \wt{K}^{+} (\chi_{x}-\chi_{u})\text{.}$$

	For simplicity, let 
$$\beta:= \frac{\epsilon^2}{ c_2\cdot \log n}.$$
		By~\eqref{eq:i23g93hg}  we have
		\begin{equation}\label{eq:2i3gbibg}
		\begin{split}
		|\wt{\mathbf{f}}_{xv}(uv)|&\ge \frac{b_{uv}^\top \wt{K}^+ b_{uv}}{2}\\
		&\ge \frac{5}{12}\cdot ||M\mathbf{b}_e||_2^2\text{~~~~~~~~~~~~~By assumption {\bf (B)} of the lemma}\\
		&\ge \frac{1}{3}\cdot \frac{1}{2^{s+1}\cdot c_2}\cdot \frac{\eps^2}{\log n} \text{~~~~~~Since $p_e' \ge \frac{1}{2^{s+1}}$}\\
		&=\frac{1}{3}\cdot \frac{\beta}{2^{s+1}}
		\end{split}
		\end{equation}
		
		Since $x,v$ belong to the same cell, by Claim~\ref{lem:celldiambound}, $||M(\chi_u-\chi_v)||_2\le w\cdot \sqrt{ q}$, thus,
		\begin{equation}\label{eq:f2}
		\begin{split}
		||\mathbf{\wt{f}}_{xv}||_2^2&\leq 3\cdot R^{\wt{K}}_{xv}\\
		&\le \frac{5}{4}(w^2\cdot q)\text{~~~~~~~~~~~~~~~Since $R^{\wt{K}}_{xv}\leq \frac{15}{4}||M(\chi_x-\chi_v)||_2^2$ by {\bf (B)}}\\
		&=15q^3\cdot\frac{\e^2}{c_2\cdot2^s\cdot\log n}\text{~~~~By line~\ref{line:w-set} of Algorithm~\ref{alg:main-heavy-edge}}\\
		&=15q^3\cdot \frac{\beta}{2^s}
		\end{split}
		\end{equation}
		Now, let $\mathbf{\wt{f}}_{xv}^{(s)}:=B_s\wt{K}\mathbf{b}_{xv}$ denote an independent sample of the entries of $\mathbf{\wt{f}}_{xv}$ with probability $\frac{1}{2^s}$. We now argue that, if the edge $(u, v)$ is included in $B_s$, then it is recovered with high probability by the heavy hitter procedure \textsc{HeavyHitter} in line~\ref{line:heavy-hitter}. We let $\mathbf{\wt{f}}^{(s)}:=\mathbf{\wt{f}}_{xv}^{(s)}$ and $\mathbf{\wt{f}}:=\mathbf{\wt{f}}_{xv}$ (i.e., we omit the subscript $xv$) to simplify notation.
		
		We will prove a lower bound on $\frac{\mathbf{\wt{f}}^{(s)}(uv)^2}{||\mathbf{\wt{f}}^{(s)}||_2^2}$ that holds with high probability. Note that
		\begin{align}
		||\mathbf{\wt{f}}^{(s)}||_2^2=\sum_{e\in B_s\setminus \{(u,v)\}} \mathbf{\wt{f}}^{(s)}(e)^2+\mathbf{\wt{f}}^{(s)}\left(uv\right)^2
		\end{align}
		For ease of notation let $X:=\sum_{e\in B_s\setminus \{(u,v)\}} \mathbf{\wt{f}}^{(s)}(e)^2$, and let $\tau:=R_{xv}^{\wt{K}}$. Thus, we have for a sufficiently large constant $C>1$
		\begin{equation}\label{eq:932hg932hg}
		\begin{split}
		&\text{Pr}\left(\frac{ \mathbf{\wt{f}}^{(s)}(uv)^2}{||\mathbf{\wt{f}}^{(s)}||_2^2}< \frac{1}{C^2\cdot q^6}\cdot \frac{\epsilon^2}{ \log n}\Big| (u,v)\in B_s \right)\\
		=&\text{Pr}\left(  X > \left(\frac{C^2\cdot q^6\cdot \log n }{\eps^2}-1\right)\cdot\mathbf{\wt{f}}^{(s)}(uv)^2 \Big| (u,v)\in B_s \right)\\
		\le& \text{Pr}\left(  ||\mathbf{\wt{f}}^{(s)}||_2^2 >\frac12 \cdot C^2\cdot q^6\cdot \frac{\log n}{\epsilon^2}\cdot 
		\mathbf{\wt{f}}(uv)^2\right)\\
		\leq& \text{Pr}\left(  \left|\left|\frac{\mathbf{\wt{f}}^{(s)}}{\tau}\right|\right|_2^2 >\frac{1}{\tau^2}\cdot \frac12 \cdot C^2\cdot q^6\cdot \frac{\log n}{\epsilon^2}\cdot \mathbf{\wt{f}}(uv)^2\right)\\
				=& \text{Pr}\left(  ||\mathbf{\wt{y}}^{(s)}||_2^2 >\frac{1}{\tau^2}\cdot \frac12 \cdot C^2\cdot q^6\cdot \frac{\log n}{\epsilon^2}\cdot \mathbf{\wt{f}}(uv)^2\right),\\
		\end{split}
		\end{equation}
		where we let $\mathbf{\wt{y}}:=\frac{\mathbf{\wt{f}}}{\tau}$ and $\mathbf{\wt{y}}^{(s)}:=\frac{\mathbf{\wt{f}}^{(s)}}{\tau}$ to simplify notation in the last line and used the fact that $\mathbf{\wt{f}}^{(s)}(uv)^2=\mathbf{\wt{f}}(uv)^2$ conditioned on $(u, v)\in B_s$ in going from line~2 to line~3. Noting that $|\mathbf{\wt{f}}(uv)|\geq \frac{1}{3}\cdot \frac{\beta}{2^{s+1}}$ by~\eqref{eq:2i3gbibg} and $\tau\leq 5q^3\cdot \frac{\beta}{2^s}$ by~\eqref{eq:f2}, we get that the last line in~\eqref{eq:932hg932hg} is upper bounded by 
		\begin{equation}\label{eq:239hg0932g}
		\text{Pr}\left(\frac{ \mathbf{\wt{f}}^{(s)}(uv)^2}{||\mathbf{\wt{f}}^{(s)}||_2^2}< \frac{1}{C^2\cdot q^6}\cdot \frac{\epsilon^2}{ \log n}\Big| (u,v)\in B_s \right)\leq \text{Pr}\left(  ||\mathbf{\wt{y}}^{(s)}||_2^2 >\frac{C'\cdot \log n}{\epsilon^2}\right),
		\end{equation}
where $C'$ is a constant that can be made arbitrarily large by increasing $C$.
			On the other hand, we have the following
		\begin{align*}
		\mathbf{E}\left( ||\mathbf{\wt{y}}^{(s)}||_2^2\right) &=\frac{1}{2^s}\cdot\frac{||\mathbf{\wt{f}}||_2^2}{\tau^2} &&\text{Since $\mathbf{\wt{f}}^{(s)}$ is obtained by sampling at rate }\frac{1}{2^s} \\
		&\le \frac{1}{2^s}\cdot \frac{3}{\tau}&&\text{By \eqref{eq:f2}}\\
		&\le \frac{1}{2^s}\cdot\frac{6}{R_{uv}^{\wt{K}}}&& \text{By \eqref{eq:i23g93hg} and Fact~\ref{fact:maxphi2}} \\
		&\le \frac{1}{2^s}\cdot \frac{36}{5 \cdot ||M\mathbf{b}_{uv}||_2^2} &&\text{By assumption {\bf (B)} of the lemma}\\
		&\le \frac{1}{2^s}\cdot \frac{36}{5 \cdot \frac{4}{5\cdot c_2}\cdot \frac{1}{2^{s+1}}\cdot \frac{\epsilon^2}{\log n}}&&\text{By condition {\bf (1)} of the lemma}\\
		&= \frac{18\cdot c_2\cdot \log n}{\epsilon^2},
		\end{align*}
		where the transition from line~2 to line~3 is justified by noting that 
		$$
		\tau\geq \left|(\chi_{u}-\chi_{v})^\T \wt{K}^{+} (\chi_{x}-\chi_{v})\right|\geq \frac1{2}(\chi_{u}-\chi_{v})^\T \wt{K}^{+} (\chi_{u}-\chi_{v})=\frac1{2}R_{uv}^{\wt{K}}
		$$
		by Fact~\ref{fact:maxphi2} and choice of $x$.
		
		We now upper bound the right hand side of~\eqref{eq:239hg0932g}. For every entry $(a, b)$ in $\mathbf{\wt{y}}$, using Fact~\ref{fact:maxphi2} one has 
		$$
		\left|\mathbf{\wt{y}}_{ab}\right|=\frac{|(\chi_{a}-\chi_{b})^\T \wt{K}^{+} (\chi_{x}-\chi_{v})|}{\tau}\leq \frac{|(\chi_{x}-\chi_{v})^\T \wt{K}^{+} (\chi_{x}-\chi_{v})|}{\tau}= 1
		$$
		
		Thus, every entry 
		 is in $[-1,1]$, and since every entry is sampled independently, so we can use standard Chernoff/Hoeffding \cite{doi:10.1080/01621459.1963.10500830} bound and we get 
		\begin{align*}
		\text{Pr}\left(  \left|\left|\mathbf{\wt y}^{(s)}\right|\right|_2^2 >\frac{C'\cdot \log n}{\epsilon^2}\right)\leq n^{-10}
		\end{align*}
		as long as $C'$ is a sufficiently large absolute constant (which can be achieved by making the constant $C$ sufficiently large). Hence, we get from~\eqref{eq:932hg932hg} that with high probability over the sampling of entries in $B_s$ 
		\begin{align*}
		\frac{|\wt{\mathbf{f}}^{(s)}(uv)|}{||\wt{\mathbf{f}}^{(s)}||_2}\geq \frac{1}{C\cdot q^3}\cdot \sqrt{\frac{\epsilon^2}{ \log n}}.
		\end{align*}
		
		We set $\eta=\frac{1}{2} \cdot \frac{1}{C\cdot q^3}\cdot \sqrt{\frac{\epsilon^2}{ \log n}}$, thus if $|\wt{\mathbf{f}}^{(s)}(uv) |\geq 2\eta ||\wt{\mathbf{f}}^{(s)}||_2$ our sparse recovery sketch must return $uv$ with high probability, by Lemma~\ref{lem:HH}.
\end{proofof}

\begin{theorem}\label{lem:corr-main-HE} (Correctness of Algorithm~\ref{alg:main-sparsify})
	Algorithm \textsc{Sparsify$(\Pi^{\le \ell} B,\ell,\epsilon)$}, for $\ell=d+1=\lceil\log_2 \frac{\lambda_u}{\lambda_\ell}\rceil+1$ (see Lemma~\ref{lem:chain_coarse}), any $\epsilon\in (0,1/5)$ and sketches $\Pi^{\le \ell} B$ of graph $G$ as described in Section~\ref{sec:sketches}, returns a graph $H$ with $O(n\cdot\polylog n\cdot\e^{-2})$ weighted edges, with Laplacian matrix $L_H$, such that $$L_H\approx_\e L_G,$$ with high probability.
\end{theorem}
\begin{proof}

	Let $\gamma=\lambda_u/2^\ell$. As the algorithm only makes recursive calls with lower values of $\ell$, we proceed by induction on $\ell$.
\paragraph{Inductive hypothesis:} A call of \textsc{Sparsify$(\Pi^{\le\ell}B,\ell,\e)$} returns a graph $H^\ell$ with $O(n\cdot\polylog n\cdot\e^{-2})$ weighted edges, with Laplacian matrix $L_{H^\ell}$, such that $$L_{H^\ell}\approx_\e L_{\ell}$$ with high probability, where
$$L_{\ell}=\begin{cases}
L_G+\frac{\lambda_u}{2^\ell}I & \text{if $0\le\ell\le d$}\\
L_G & \text{if $\ell=d+1.$}
\end{cases}$$
and for all $\ell\geq 0$ the matrix $\wt{K}$ defined at the beginning of Algorithm~\ref{alg:main-sparsify} is a $3$-spectral sparsifier of $G^{\gamma(\ell)}$. 
		\paragraph{Base case: $\ell=0$.}  In this case we set $\wt{K}=\lambda_u I$ (see line \ref{line:leaf-case-K} of Algorithm \ref{alg:main-sparsify}). By Lemma~\ref{lem:chain_coarse} we have  
	\begin{align}\label{eq:l0case}
	\frac{1}{2}\cdot
	L_\ell \preceq \wt{K}  \preceq L_\ell,
	\end{align}
i.e. $\wt{K}$ is a factor $3$ spectral approximation of $L_\ell$ for $\ell=0$. We argue that the graph output by Algorithm~\ref{alg:main-sparsify} satisfies $L_{H^\ell}\approx_\e L_{\ell}$ below, together with the same argument for the inductive step.

	\paragraph{Inductive step: $\ell-1\to \ell$.} 	As per line~\ref{line:HE-CallSp} of Algorithm~\ref{alg:main-sparsify} we set $\wt{K}=\frac{1}{2(1+\e)}\textsc{Sparsify}(\Pi^{\le\ell-1}B,\ell-1, \e)$, therefore the corresponding Laplacian for this call is $L_{\ell-1}$. By the inductive hypothesis \textsc{Sparsify}$(\Pi^{\le\ell-1}B,\ell-1, \e)$ returns an $\e$-sparsifier of $L_{\ell-1}$, so we have 
	\begin{equation}
	\label{eq:mn_ep-ap}
	(1-\e) \cdot L_{\ell-1} \preceq_r 2(1+\e)\wt{K} \preceq_r (1+\e) \cdot L_{\ell-1}\text{.}
	\end{equation}
	Moreover, by Lemma~\ref{lem:chain_coarse}, we have 
	\begin{equation}
	\label{eq:mn_Gamma-ap}
	\frac{1}{2} \cdot L_\ell \preceq \frac{1}{2} \cdot L_{\ell-1} \preceq L_\ell \text{.}
	\end{equation}
	Putting \eqref{eq:mn_ep-ap} and \eqref{eq:mn_Gamma-ap} together we get
	\begin{equation}\label{eq:LK}
	\frac{1-\e}{2(1+\e)} \cdot L_\ell\preceq_r  \wt{K} \preceq_r  L_\ell,
	\end{equation}
	which implies for $\e\le1/5$ that
	\begin{align}\label{eq:lge0case}
	\frac13 \cdot L_\ell \preceq_r \wt{K} \preceq_r  L_\ell \text{.}
	\end{align}

We thus have that for all values of $\ell$ the matrix $\wt{K}$ defined at the beginning of Algorithm~\ref{alg:main-sparsify} is a $3$-spectral sparsifier of $G^{\gamma(\ell)}$, assuming the inductive hypothesis for $\ell-1$ (except for the base case case, where no inductive hypothesis is needed).   Consequently, for any pair of vertices $(u,v)$ in the same connected component in $L$, 
	\begin{align}\label{eq:effCineq}
	b_{uv}^\top L_\ell^+ b_{uv} \le  b_{uv}^\top \wt{K}^+ b_{uv} \le 3\cdot b_{uv}^\top L_\ell^+ b_{uv}
	\end{align}

	For the rest of the proof, we let $L:=L_\ell$ for simplicity. We now show that the rest of the algorithm constructs an $\e$-sparsifier for $L$ by sampling each edge $e$ with some probability at least $\min\{1,R^{L}_{uv}\log(n)/\e^2\}$ and giving it weight inverse proportional to the probability. This will indeed give us an $\e$-sparsifier due to Theorem~\ref{lem:classic_result}. In particular, this probability will be the following: For edges $e$ in the appended complete graph $\gamma I$ the probability is $1$. For an edge $e$ in the original graph $G$ we define the variable $p_e'$, as in line~\ref{def:p_e} of Algorithm~\ref{alg:main-heavy-edge}, to be $\frac54\cdot c_2\cdot||Mb_e||_2^2\cdot\log(n)/\e^2$, and we define $p_e$ to be $\min\{1,p_e'\}$. Let $s_e$ be the integer such that $p_e'\in(2^{-s_e-1}, 2^{-s_e}]$. Note that then $p_e\in(2^{-s_e^+-1},2^{-s_e^+}]$. Our probability for sampling an edge $e$ of the original graph will be $2^{-s_e^+}$, which is less than $\min\{1,c_2\cdot R^{L}_{e}\log(n)/\e^2\}$, as required by Theorem~\ref{lem:classic_result}.

Consider the conditions of Lemma~\ref{lm:main}.

\begin{enumerate}
	\item  $\frac13 \cdot L_\ell \preceq_r \wt{K} \preceq_r  L_\ell$ is satisfied as shown above.
	\item Note that $R^{\wt K}_{uv}=||\wt{W}^{1/2} \wt B\wt K^+b_u-\wt{W}^{1/2} \wt B\wt K^+b_v||^2_2$ so we can use the Johnson-Lindenstrauss lemma to approximate $R^{\wt K}_{uv}$ using a smaller matrix. In lines~\ref{line:setq} and \ref{line:setQ} we use the exact construction of Lemma~\ref{lem:JL} with $q$ being large enough for parameters $\e=1/5$ and $\beta=6$. Therefore,  $R_{uv}^{\wt{K}}\le\frac{5}{4}||M(\chi_u-\chi_v)||_2^2\le \frac{3}{2} R_{uv}^{\wt{K}}$ is satisfied with high probability, by Lemma~\ref{lem:JL}.
\end{enumerate}

Thus by Lemma~\ref{lm:main} if edge $e$ is sampled in $B^\ell_{s_e^+}$ then  \textsc{RecoverEdges($\Pi^\ell_{s_e^+}B^\ell_{s_e^+},M,\wt K^+,s_e,q,\e$)} will recover $e$ with high probability in line~\ref{line:recover} of Algorithm~\ref{alg:main-sparsify}. It will then be given the required weight ($2^{s_e^+}$). Note that $e$ will not be recovered in any other call of \textsc{RecoverEdges}, that is when $s\neq s_e$. Note also, that $2^{-s_e^+}$ is indeed an upper bound on $\min\{1,c_2\cdot R^{L}_{e}\cdot \log(n)/\e^2\}$, and within constant factor of it. Therefore, by Theorem~\ref{lem:classic_result}, the resulting graph will be an $\e$-spectral sparsifier of $G^\gamma$, and it will be $O(n\cdot \polylog(n)/\e^2)$-sparse (disregarding the regularization).
\end{proof}

\subsection{Maintenance of sketches}\label{sec:sketches}

Note that Algorithm \ref{alg:main-heavy-edge} takes sketch $\Pi B$ as input. More precisely, $\Pi$ is a concatenation of \textsc{HeavyHitter} sketch matrices composed with sampling matrices, indexed by sampling rate $s$ and regularization level $\ell$. In particular, for all $s$ and $\ell$ let $B^\ell_s$ be a row-sampled version of $B$ at rate $2^{-s}$. Then $\Pi^\ell_s$ is a \textsc{HeavyHitter} sketch drawn from the distribution from Lemma~\ref{lem:HH} with parameter $\eta=\frac{1}{2} \cdot \frac{1}{C\cdot q^3}\cdot \sqrt{\frac{\epsilon^2}{ \log n}}$. Note that the matrices $\left(\Pi^\ell_s\right)_{s,\ell}$ are independent and identically distributed. We then maintain $\Pi^\ell_s B^\ell_s$ for all $s$ and $\ell$. We define 
$$\Pi^\ell B=\Pi^\ell_0 B^\ell_0\oplus\ldots\oplus\Pi^\ell_{10\log n}B^\ell_{10\log n},
$$
where $\oplus$ denotes concatenation of rows. We let $\Pi^{\le \ell}$ denote $\Pi^{0}\oplus\ldots\oplus\Pi^{\ell}$, and let  $\Pi$ denote $\Pi^{\le d+1}$ to simplify notation. Thus, the algorithm maintains $\Pi B$ throughout the stream. We maintain $\Pi B$ by maintaining each $\Pi^\ell_s B^\ell_s$ individually. To this end we have for each $s$ and $\ell$ an independent hash function $h^\ell_s$ mapping ${V\choose 2}$ to $\{0,1\}$ independently such that $\mathbb P(h^\ell_s(u,v)=1)=2^{-s}$. Then when an edge insertion or deletion, $\pm(u,v)$, arrives in the stream, we update $\Pi^\ell_s B^\ell_s$ by $\pm\Pi^\ell_s\cdot b_{uv}\cdot h^\ell_s(u,v)$.

Overall, the number of random bits needed for all the matrices in an invocation of Algorithm~\ref{alg:main-heavy-edge} is at most $R=\wt{O}(n^2)$, in addition to the random bits needed for the recursive calls. To generate matrix $\Pi$ we use the fast pseudo random numbers generator from Theorem~\ref{thm:prg} below:

\begin{theorem}\cite{KMMMN19}
	\label{thm:prg}
	For any constants $q,c > 0$, there is an explicit pseudo-random generator (PRG) that draws on a seed of $O(S\plog(S))$ random bits and can simulate any randomized algorithm running in space $S$ and using $R = O(S^q)$ random bits. This PRG can output any pseudorandom bit in $O(\log^{O(q)} S)$ time and the simulated algorithm fails with probability at most $S^{-c}$ higher than the original.
\end{theorem}

 Observe that the space used by Algorithm~\ref{alg:main-heavy-edge} is $s=\wt{O}(n)$  in addition to the space used by the recursive calls. Since $R=O(n^2)$, we have $R=O(s^2)$. Therefore, by Theorem~\ref{thm:prg} we can generate seed of $O(s\cdot \poly(\log s))$ random bits in $O(s\cdot \poly(\log s))$ time that can simulate our randomized algorithm.

Also, note that the random matrix $Q \in\mathbb{R} ^{\Theta(\log n)\times {n \choose 2}}$ for JL (line~\ref{line:setQ} of Algorithm \ref{alg:main-heavy-edge}) can be generated using $O(\log n)$-wise independent hash functions.

\subsection{Proof of Theorem~\ref{thm:main}}\label{sec:proof}

\begin{proofof}{Theorem~\ref{thm:main}}

Correctness of Algorithm~\ref{alg:main-heavy-edge} is proved in Theorem~\ref{lem:corr-main-HE}. It remains to prove space and runtime bounds.
	
\paragraph{Run-time and space analysis.} We will prove that one call of \textsc{Sparsify} in Algorithm~\ref{alg:main-sparsify} requires $\wt{O}(n)$ time and space, discounting the recursive call, where $n$ is the size of the vertex set of the input graph. Consider first lines~\ref{line:setQ} and \ref{line:setM}, and note that the random matrix $Q \in\mathbb{R} ^{\Theta(\log n)\times {n \choose 2}}$ for JL (line~\ref{line:setQ} of Algorithm \ref{alg:main-heavy-edge}) can be generated using $O(\log n)$-wise independent hash functions, resulting in $\poly(\log n)$  time to generate an entry of $Q$ and $O(\log n)$ space. We then multiply $Q\wt{W}^{1/2}\wt B$ by $\wt K^+$ which amounts to solving $\Theta(\log n)$ Laplacian systems and can be done in $O(n\polylog n\cdot\e^{-2})$ time, since $\wt K$ is $O(n\polylog(n)\cdot\e^{-2})$ sparse, using any of a variety of algorithms in the long line of improvements in solving Laplacian systems \cite{SpielmanT04,KoutisMP10,KoutisMP11,KelnerOSZ13,LeeS13,PengS14,CohenKMPPRX14,KoutisLP16,KyngLPSS16,KyngS16}. The resulting matrix, $M$, is again $\Theta(\log n\times n)$ and can be stored in $n\polylog n$ space. We note that the aforementioned Laplacian solvers provide approximate solutions with inverse polynomial precision, which is sufficient for application of the \textsc{HeavyHitter} sketch.

The for loops in both line~\ref{line:s-for} and line~\ref{line:j-for} iterate over only $\Theta(\log n)$ values. For all non-empty cells we iterate over all vertices in that cell, so overall, we iterate $n$ times. The \textsc{HeavyHitters} subroutine called with parameter $\eta=\e/\polylog n$ returns by definition at most $\polylog n/\e^2$ elements, so the for loop in line~\ref{line:e-for} is over $\polylog n/\e^2$ iterations. In total this is $O( n\polylog n \cdot \e^{-2})$ time and space as claimed.

To get an $\e$-sparsifier of the input graph $G$, we need only to run \textsc{Sparsify}$(\Pi^{\le d+1}B, d+1,\e)$. Therefore chain of recursive calls will be $\Theta(\log(n))$ long, and the total run time will still be $\wt{O}(n\e^{-2})$.
\end{proofof}

\newcommand{\etalchar}[1]{$^{#1}$}

\appendix
\section{Supplementary material}
\label{app:algo}
We will need Lemma \ref{lem:chain_coarse} that we use in the correctness proof of our algorithm. 

\begin{lemma}[Chain of Coarse Sparsifiers \cite{LiMP13, kapralov2017single}]
	\label{lem:chain_coarse}
	Consider any PSD matrix $K$ with
	maximum eigenvalue bounded from above by $\lambda_u=2n$ and minimum nonzero eigenvalue bounded from below by $\lambda_\ell=\frac{1}{8n^2}$. Let $d = \ceil{\log_2 \frac{\lambda_u}{\lambda_\ell}}$. For $\ell \in \{0, 1, 2, \ldots, d\}$, define: $$\gamma(\ell)=\frac{\lambda_u}{2^\ell}.$$ So $\gamma(d)\leq \lambda_\ell$, and $\gamma(0)=\lambda_u$. Then the chain of $PSD$ matrices, $[K_0, K_1, \ldots, K_d]$ with $K_\ell=K+\gamma(\ell)I$ satisfies the following relations:
	\begin{enumerate}
		\item $K \preceq_r K_d \preceq_r 2 \cdot K, $
		\label{itm:last}
		\item $K_\ell \preceq K_{\ell-1} \preceq 2 \cdot K_\ell $ for all $\ell \in \{1,\ldots, d\},$\label{itm:mid}
		\item $K_0 \preceq 2\cdot \gamma(0) \cdot I \preceq 2 \cdot K_0.$ \label{itm:base}
	\end{enumerate}
\end{lemma}

We will need Theorem~\ref{lem:classic_result} that we use in the proof of correctness of the main algorithm. It is well known that by sampling the edges  of $B$ according to their effective resistance, it is possible to obtain a weighted edge vertex incident matrix $\wt{B}$ such that $(1-\e)B^\top B \preceq \wt{B}^\top\wt{B} \preceq (1+\e)B^\top B$ with high probability (see Lemma \ref{lem:classic_result}).

\begin{theorem}[Spectral Approximation via Effective Resistance Sampling \cite{spielman2011graph}] 
	\label{lem:classic_result}
	Let $B\in \R^{{{n}\choose{2}}\times n}$, $K=B^\top B$, and let $\wt{\tau}$ be a vector of leverage
	score overestimates for $B$'s rows, i.e. $\wt{\tau}_y \geq  \mathbf{b}_y^\top K^{+} \mathbf{b}_y$ for all $y \in [m]$. For $\e \in (0,1)$ and fixed constant $c$, define the sampling probability for row $\mathbf{b}_y$ to be $p_y = \min\{1, c\cdot \e^{-2}\log n \cdot\wt{\tau}_y \}$. Define a diagonal sampling matrix $W$ with $W(y,y) = \frac{1}{p_y}$ with probability $p_y$ and $W(y,y) = 0$ otherwise. With high probability,
	$\wt{K}=B^\top WB\ape K$. Furthermore $W$ has $O(||\wt{\tau}||_1\cdot \e^{-2}\log n )$ non-zeros with high probability.
\end{theorem}


\label{subsec:hv-hit}
\begin{lemma}[$\ell_2$ Heavy Hitters]\label{lem:HH} For any $\eta > 0$, there is a decoding algorithm denoted by \textsc{HeavyHitter} and a distribution on matrices $S^h$ in $\R^{O(\eta^{-2} \polylog(N))\times N}$ such that, for any $x \in \R^N$, given $S^h x$, the algorithm $\textsc{HeavyHitter}(S^h x, \eta)$ returns a list $F\subseteq [N]$ such that $|F|=O(\eta^{-2} \polylog(N))$	with probability $1 - \frac{1}{\poly(N)}$ over the choice of $S^h$ one has 

\begin{description}
\item[(1)] for every $i\in [N]$ such that $|x_i|\geq \eta ||x||_2$ one has $i\in F$;
\item[(2)] for every $i\in F$ one has $|x_i|\geq (\eta/2) ||x||_2$.
\end{description}
The sketch $S^hx$ can be maintained and decoded in $O(\eta^{-2} \polylog(N))$ time and space.
\end{lemma}

\begin{lemma}[Binary Johnson-Lindenstrauss Lemma \cite{DBLP:journals/jcss/Achlioptas03}]{\label{lem:JL}}

Let $P$ be an arbitrary set of points in $\mathbb R^d$, represented by a $d\times n$ matrix $A$, such that the $j^\text{th}$ point is $A\chi_j$. Given $\e,\ \beta>0$ and
$$q\ge\frac{4+2\beta}{\epsilon^2/2-\epsilon^3/3}\log n.$$
Let $Q$ be a random $q\times d$ matrix $(q_{ij})_{ij}$ where $q_{ij}$'s are independent identically distributed variables taking $1$ and $-1$ each with probability $1/2$. Then, if $M=\frac1{\sqrt{q}}QA$, then with probability at least $1-n^{-\beta}$, for all $u,v\in[n]$
$$(1-\e)||A\chi_u-A\chi_v||^2_2\le||M\chi_u-M\chi_v||^2_2\le(1+\e)||A\chi_u-A\chi_v||^2_2$$

\end{lemma}

\end{document}